\documentclass[11pt]{amsart}

\usepackage{amsmath,amssymb,amsthm,mathtools,microtype}
\usepackage[hidelinks]{hyperref}
\hypersetup{
  pdftitle={Bezoutian Decoupling for Conformal Yang--Mills Multiplets in (A)dS},
  pdfauthor={Weiqi Jiang}
}

\newtheorem{theorem}{Theorem}[section]
\newtheorem{proposition}[theorem]{Proposition}
\newtheorem{lemma}[theorem]{Lemma}
\newtheorem{corollary}[theorem]{Corollary}
\theoremstyle{remark}

\newcommand{\AdS}{(A)dS}
\newcommand{\R}{\mathbb R}

\newcommand{\one}{\mathbf 1}

\title[Bezoutian decoupling for conformal Yang--Mills]{Bezoutian Decoupling for Conformal Yang--Mills Multiplets in $(A)dS$}
\author{Weiqi Jiang}
\address{Institute of Theoretical Physics, Chinese Academy of Sciences, Beijing, China}

\begin{document}
\raggedbottom

\begin{abstract}
Metsaev exhibited generic and decoupled formulations of conformal
Yang--Mills theory in $(A)dS_6$, $(A)dS_8$, and $(A)dS_{10}$, and conjectured
the corresponding relations in higher even dimensions.  Motivated by the
low-dimensional matrices in Appendix C, we identify and diagonalize a
natural all-$N$ Bezoutian continuation of the three generic Gram forms.
The monic Hankel cutoff, the evaluation pattern in Appendix C, and the
conjectured mass nodes determine this continuation uniquely.  Its vector and
radical Gram-form diagonalization identities hold for every
$D=d+1=2N+4$, $N\geq1$.  The continuation is the Bezout matrix of
$z\prod_{i=1}^N(z-\rho i(2N+1-i))$ and $1$; evaluation at its simple roots
produces the Vandermonde congruence appearing in the field redefinition.  The
same calculation gives closed formulas for the inverse, determinant, and
inertia, and identifies the normalization weights with Johnson-graph
multiplicities.  Any nonlinear coefficient algebra already specified in the
generic formulation is transported by this change of basis.  The result does
not construct such an algebra in arbitrary dimension.  At $N=4$, we impose
associativity, Frobenius invariance, and a fixed flat specialization; these
algebraic constraints still admit a one-parameter family of pairwise distinct
products in a fixed generic basis.
\end{abstract}

\maketitle
\tableofcontents

\section{Introduction}\label{sec:introduction}

Ordinary-derivative formulations replace the higher derivatives of a
conformal gauge theory by auxiliary vector fields and Stueckelberg scalars.
For conformal Yang--Mills theory in flat even-dimensional space this
formulation is available in arbitrary dimension~\cite{Metsaev2024Flat}.
In $\AdS$ space, Metsaev constructed two presentations in
dimensions six, eight, and ten~\cite{Metsaev2024AdS}.  The generic
presentation is regular in the flat limit but has nondiagonal kinetic terms.
The decoupled presentation separates one massless vector from a tower of
massive Stueckelberg systems.  Appendix~C of~\cite{Metsaev2024AdS} gives the
change-of-basis matrices in the three computed dimensions and states the
corresponding relations as conjectures for higher odd values of the parameter
$d$.

Motivated by those low-dimensional matrices, we identify a natural all-$N$
Bezoutian continuation and prove the associated Gram-form identities.  The
continuation retains the monic Hankel cutoff of the known generic forms, the
evaluation pattern of the Appendix~C field map, and the proposed mass
parameters.  These requirements determine it uniquely; thus, under the stated
continuation assumptions, the Gram form is forced rather than chosen to fit a
diagonal answer.  The standard simple-root Vandermonde congruence is
implemented by precisely the matrix appearing in the field redefinition.  The
normalization reduces to a factorial identity, and the same calculation yields
the inverse transformation and the inertia of the kinetic form.

Following the notation of~\cite{Metsaev2024AdS}, let $d$ be one less than the
spacetime dimension $D$ and write
\begin{equation}\label{eq:intro-parameters}
 d=2N+3,\qquad D=d+1=2N+4,\qquad N\in\mathbb Z_{>0},\qquad
 \rho\in\mathbb R\setminus\{0\},
\end{equation}
and define
\begin{equation}\label{eq:intro-nodes}
 \lambda_i=i(2N+1-i),\qquad x_i=\rho\lambda_i,\qquad 1\leq i\leq N.
\end{equation}
The sign of $\rho$ distinguishes the de Sitter and anti-de Sitter
conventions.  All square roots below are taken on the positive real branch.

\begin{theorem}[Bezoutian Gram-form diagonalization]\label{thm:main}
Let $H_N$ and $H_N^{\rm sub}$ be the vector and radical Gram matrices defined
in Section~\ref{sec:formulations}.  Put
\begin{equation}\label{eq:intro-weights}
 w_i=\frac{(2N)!(2N+1-2i)}{i!(2N+1-i)!},\qquad
 k_\rho=(-\rho)^N(2N)!,
\end{equation}
and let
\begin{equation}\label{eq:intro-matrices}
 (M_N)_{ij}=x_i^j,\qquad
 K_N=\operatorname{diag}_{1\leq i\leq N}((-1)^i\sqrt{w_i}),
\end{equation}
with $L_{N+1}$ as in~\eqref{eq:L-definition}.  Then $M_N$ and $L_{N+1}$ are
invertible and
\begin{align}
 L_{N+1}H_NL_{N+1}^{\mathsf T}
 &=k_\rho\operatorname{diag}(1,-1,1,\ldots,(-1)^N),
 \label{eq:intro-vector-congruence}\\
 K_NM_NH_N^{\rm sub}M_N^{\mathsf T}K_N
 &=k_\rho\operatorname{diag}_{1\leq i\leq N}
       \bigl((-1)^i x_i\bigr).
 \label{eq:intro-radical-congruence}
\end{align}
The radical diagonal entries reproduce Metsaev's conjectured mass nodes,
\begin{equation}\label{eq:intro-masses}
 m_i^2=x_i=\rho i(d-2-i),\qquad 1\leq i\leq N.
\end{equation}
The determinant, inverse, and inertia are given explicitly in
Section~\ref{sec:inverse}.
\end{theorem}

The mass sequence itself was known from the free conformal-field
factorization~\cite{Metsaev2014}.  Theorem~\ref{thm:main} provides an all-$N$
Vandermonde realization of the map pattern displayed in Appendix~C of
\cite{Metsaev2024AdS} for $N\leq3$.  It proves the corresponding Gram
congruences for the uniquely characterized continuation, including the
normalization, radical map, and inverse, in every dimension covered
by~\eqref{eq:intro-parameters}.

An invertible change of basis also transports any multiplication, ideal, gauge
bracket, and invariant form already present in the generic formulation; this
is proved in Section~\ref{sec:transport}.  It does not construct the
conjectural source dynamics for arbitrary $N$ or prove the separate
generic-basis equations that were not derived in~\cite{Metsaev2024AdS}.

The diagonal free action and the invertible field map also fix the two
generic quadratic forms uniquely.  Conditional on the proposed interacting
action and pairing, equation (2.43) then follows by a finite regrouping for
every $N$.  At $N=4$, the first case not treated in
\cite{Metsaev2024AdS}, the natural associative Frobenius conditions with
fixed flat specialization leave a one-parameter family of pairwise distinct
products in the fixed generic basis.  Thus the linear data do not select a
unique nonlinear continuation.  These statements are proved in
Section~\ref{sec:quadratic-determination}.

A separately archived reproducibility supplement contains exact symbolic
checks~\cite{Jiang2026Zenodo}.  They cover the Gram congruences, the published
low-dimensional matrices, and the $N=4$ family.

The proof is organized as follows.  Section~\ref{sec:formulations} defines the
two source forms.  Sections~\ref{sec:spectral-polynomial} and
\ref{sec:decoupling} prove the Bezoutian congruences.  We derive the inverse
and low-dimensional specializations in Section~\ref{sec:inverse}, explain
nonlinear transport in Section~\ref{sec:transport}, and separate the
determined quadratic data from the nonlinear ambiguity in
Section~\ref{sec:quadratic-determination}.  Johnson-graph, $SU(2)$, and
Legendre--Stirling models appear in Section~\ref{sec:spectral-models}; the
remaining questions are discussed in Section~\ref{sec:discussion}.

\section{The generic forms}\label{sec:formulations}

Define the spectral polynomial
\begin{equation}\label{eq:Q-definition}
 Q_N(z)=\prod_{i=1}^N(z-x_i)=\sum_{r=0}^N c_r z^r.
\end{equation}
The vector Gram form in the generic basis $E_0,\ldots,E_N$ is the Hankel
matrix
\begin{equation}\label{eq:H-definition}
 (H_N)_{mn}=
 \begin{cases}
 c_{m+n},&m+n\leq N,\\
 0,&m+n>N,
 \end{cases}
 \qquad 0\leq m,n\leq N.
\end{equation}
The coefficients are
\begin{equation}\label{eq:source-G}
 c_s=(-\rho)^{N-s}G_s,\qquad
 G_s=e_{N-s}(\lambda_1,\ldots,\lambda_N),
\end{equation}
where $e_r$ is the elementary symmetric polynomial.  Thus
\eqref{eq:H-definition} is a natural monic Hankel continuation to arbitrary
$N$ of the generic bilinear forms displayed in~\cite{Metsaev2024AdS}.  The
following characterization explains in what sense the continuation is
forced.

\begin{lemma}[Uniqueness of the Hankel continuation]
\label{lem:unique-hankel-continuation}
Fix the proposed nodes $x_1,\ldots,x_N$ and let
\[
 A(z)=\sum_{r=0}^N a_rz^r,\qquad a_N=1.
\]
Suppose that an all-$N$ continuation retains the Hankel cutoff
\[
 (\widetilde H_N)_{mn}=
 \begin{cases}
  a_{m+n},&m+n\leq N,\\
  0,&m+n>N,
 \end{cases}
\]
and the evaluation pattern of Appendix~C.  More precisely, with
$v(z)=(1,z,\ldots,z^N)$, assume that the massless evaluation row is
orthogonal to the proposed massive rows:
\[
 v(0)\widetilde H_Nv(x_i)^{\mathsf T}=0,
 \qquad 1\leq i\leq N.
\]
Then $A=Q_N$ and $\widetilde H_N=H_N$.
\end{lemma}

\begin{proof}
The zeroth row of $\widetilde H_N$ is $(a_0,\ldots,a_N)$, and hence
\[
 v(0)\widetilde H_Nv(x_i)^{\mathsf T}=A(x_i)=0.
\]
The nodes are pairwise distinct.  Since $A$ is monic of degree $N$ and
vanishes at all of them,
\[
 A(z)=\prod_{i=1}^N(z-x_i)=Q_N(z).
\]
The equality of the two Hankel matrices follows coefficient by coefficient.
\end{proof}

The lemma does not derive its continuation assumptions from an interacting
field theory.  It shows that the conjectured mass nodes and the
Appendix~C evaluation map select a unique Gram form within the Hankel class
suggested by the three known dimensions.

The Stueckelberg and auxiliary vectors belonging to the radical use the
shifted form
\begin{equation}\label{eq:Hsub-definition}
 (H_N^{\rm sub})_{mn}=c_{m+n+1},\qquad 0\leq m,n\leq N-1,
\end{equation}
with $c_s=0$ for $s>N$.

It is useful to include the massless node $x_0=0$.  Let
\begin{equation}\label{eq:V-definition}
 (V_N)_{ij}=x_i^j,\qquad 0\leq i,j\leq N,
\end{equation}
and put
\begin{equation}\label{eq:D-definition}
 D_N=\operatorname{diag}(1,-\sqrt{w_1},\sqrt{w_2},\ldots,
 (-1)^N\sqrt{w_N}).
\end{equation}
The field-redefinition matrix is
\begin{equation}\label{eq:L-definition}
 L_{N+1}=D_NV_N=
 \begin{pmatrix}1&0\\0&K_N\end{pmatrix}
 \begin{pmatrix}1&0\\\one&M_N\end{pmatrix}.
\end{equation}
The radical transformation is the lower-right evaluation block $K_NM_N$.

The three coefficient rows in the source paper are recovered immediately:
for $N=1,2,3$, the nontrivial $G_s$ are respectively
\begin{equation}\label{eq:small-G}
 (2,1),\qquad (24,10,1),\qquad (720,252,28,1).
\end{equation}
The proof below treats $N$ symbolically and uses these rows only as a check of
conventions.

\section{The spectral polynomial and its Bezoutian}\label{sec:spectral-polynomial}

Let $v(z)=(1,z,\ldots,z^N)$.  Grouping the entries of the Hankel form by
their total degree gives
\begin{align}
 v(x)H_Nv(y)^{\mathsf T}
 &=\sum_{r=0}^Nc_r\sum_{m=0}^r x^m y^{r-m}\notag\\
 &=\frac{xQ_N(x)-yQ_N(y)}{x-y}.
 \label{eq:bezout-kernel}
\end{align}
The expression on the right is understood by continuity when $x=y$.  Thus
$H_N$ is the Bezout matrix of $P_N(z)=zQ_N(z)$ and $1$.  The usual
simple-root Vandermonde reduction of a Bezoutian
\cite{HelmkeFuhrmann1989} will now be evaluated with all normalizations kept
explicit.

\begin{lemma}\label{lem:root-orthogonality}
The $N+1$ row vectors $v(x_i)$, $0\leq i\leq N$, are mutually orthogonal for
$H_N$.  Their norms are
\begin{equation}\label{eq:root-norms}
 v(0)H_Nv(0)^{\mathsf T}=k_\rho,
 \qquad
 v(x_i)H_Nv(x_i)^{\mathsf T}=x_iQ_N'(x_i),\quad i\geq1.
\end{equation}
\end{lemma}

\begin{proof}
For two distinct roots of $P_N$, equation~\eqref{eq:bezout-kernel} vanishes.
On the diagonal, differentiating the divided difference gives
$x_iQ_N'(x_i)$ for $i>0$.  At zero the norm is
\[
 c_0=(-\rho)^N\prod_{i=1}^N\lambda_i.
\]
The elementary product
\[
 \prod_{i=1}^N i(2N+1-i)=N!\frac{(2N)!}{N!}=(2N)!
\]
therefore gives $c_0=k_\rho$.
\end{proof}

The same polynomial already displays two useful structures.  On setting
$X=N(N+1)-z$ and reindexing $r=N-i$, one finds
\begin{equation}\label{eq:legendre-product}
 \prod_{i=1}^N(z-\lambda_i)
 =(-1)^N\prod_{r=0}^{N-1}\bigl(X-r(r+1)\bigr).
\end{equation}
The product on the right is the generalized falling factorial associated
with the Legendre--Stirling numbers of the first kind
\cite{EverittLittlejohnWellman2002}.  In particular, if
\[
 \prod_{r=0}^{N-1}(X-r(r+1))=\sum_{j=0}^N\ell_{N,j}X^j,
\]
then
\begin{equation}\label{eq:legendre-stirling-recurrence}
 \ell_{N,j}=\ell_{N-1,j-1}-N(N-1)\ell_{N-1,j},\qquad \ell_{0,0}=1,
\end{equation}
and the generic coefficients are the binomial translates
\begin{equation}\label{eq:G-legendre-stirling}
 G_s=\sum_{j=s}^N\binom{j}{s}\ell_{N,j}
      \bigl(N(N+1)\bigr)^{j-s}.
\end{equation}
These identities are not needed for the congruence proof, but they replace
the dimension-by-dimension $G_s$ by a classical triangular array.

There is also an operator form.  The Legendre differential operator
\[
 \mathcal L=-\frac{d}{dt}\left((1-t^2)\frac{d}{dt}\right)
\]
has eigenvalues $r(r+1)$ on the Legendre polynomials.  Hence
\eqref{eq:legendre-product} is the characteristic polynomial of
$N(N+1)I-\mathcal L$ on $\R[t]_{\leq N-1}$.  The Gram coefficients are
therefore translated characteristic coefficients of a finite Legendre
spectrum.

\section{Proof of the decoupling theorem}\label{sec:decoupling}

It remains to evaluate the diagonal norms in
Lemma~\ref{lem:root-orthogonality}.  The numbers $\lambda_i$ are strictly
increasing; hence the nodes $x_i=\rho\lambda_i$ are pairwise distinct whenever
$\rho\ne0$.  Indeed,
\begin{equation}\label{eq:node-gap}
 \lambda_{i+1}-\lambda_i=2(N-i)>0.
\end{equation}
For $i\neq j$,
\[
 \lambda_j-\lambda_i=(j-i)(2N+1-i-j).
\]
Splitting the product at $j=i$ yields
\begin{equation}\label{eq:node-derivative-absolute}
 \lambda_i\left|\prod_{j\ne i}(\lambda_i-\lambda_j)\right|
 =\frac{i!(2N+1-i)!}{2N+1-2i}.
\end{equation}
The derivative $Q_N'(x_i)$ has the sign of
$(-1)^{N-i}\rho^{N-1}$.  Combining this sign,
\eqref{eq:node-derivative-absolute}, and~\eqref{eq:intro-weights} gives the
normalization identity
\begin{equation}\label{eq:weight-derivative}
 w_i x_iQ_N'(x_i)=(-1)^i k_\rho.
\end{equation}
Notice that $2N+1-2i>0$, so every $w_i$ is positive.

\begin{proof}[Proof of the vector identity in Theorem~\ref{thm:main}]
The rows of $V_N$ are the root-evaluation vectors in
Lemma~\ref{lem:root-orthogonality}.  Consequently $V_NH_NV_N^{\mathsf T}$
is diagonal.  Its zeroth entry is $k_\rho$, and multiplication of its $i$th
entry by the two diagonal factors $(-1)^i\sqrt{w_i}$ gives
$w_i x_iQ_N'(x_i)=(-1)^ik_\rho$.  Since $L_{N+1}=D_NV_N$, this proves
\eqref{eq:intro-vector-congruence}.
\end{proof}

For the radical form, put $u(z)=(z,z^2,\ldots,z^N)$.  The definitions of
$H_N$ and $H_N^{\rm sub}$ give
\begin{equation}\label{eq:radical-kernel}
 u(x)H_N^{\rm sub}u(y)^{\mathsf T}
 =xy\,\frac{Q_N(x)-Q_N(y)}{x-y}.
\end{equation}
Indeed, both sides equal
$xy\sum_{r=1}^Nc_r\sum_{m=0}^{r-1}x^my^{r-1-m}$.
At two distinct roots of $Q_N$ the expression vanishes, while its diagonal
value is
\begin{equation}\label{eq:radical-root-norm}
 u(x_i)H_N^{\rm sub}u(x_i)^{\mathsf T}=x_i^2Q_N'(x_i).
\end{equation}
The rows of $M_N$ are $u(x_i)$.  Multiplying
\eqref{eq:radical-root-norm} by the two factors from $K_N$ and using
\eqref{eq:weight-derivative}, we obtain
\begin{equation}\label{eq:radical-congruence-proof}
 K_NM_NH_N^{\rm sub}M_N^{\mathsf T}K_N
 =k_\rho\operatorname{diag}_{1\leq i\leq N}\bigl((-1)^ix_i\bigr),
\end{equation}
which proves the radical identity and reproduces the conjectured nodes
in~\eqref{eq:intro-masses}.

It is useful to combine the massless and massive normalizations.  With
$P_N(z)=zQ_N(z)$, $x_0=0$, and $w_0=1$, equations
\eqref{eq:root-norms} and~\eqref{eq:weight-derivative} say
\begin{equation}\label{eq:barycentric-master}
 w_iP_N'(x_i)=(-1)^ik_\rho,\qquad 0\leq i\leq N.
\end{equation}
Thus the positive $w_i$ are, up to the common scale and the kinetic signs,
the reciprocal derivative weights familiar from barycentric interpolation
\cite{BerrutTrefethen2004}.

\section{Inverse, determinant, and inertia}\label{sec:inverse}

Factoring $x_i$ from the $i$th row of $M_N$ leaves an ordinary Vandermonde
matrix.  Therefore
\begin{equation}\label{eq:M-determinant}
 \det M_N=\left(\prod_{i=1}^Nx_i\right)
 \left(\prod_{1\leq i<j\leq N}(x_j-x_i)\right)\ne0.
\end{equation}
Since the squared diagonal entries of $K_N$ are the positive numbers $w_i$,
\begin{equation}\label{eq:L-determinant-square}
 (\det L_{N+1})^2=\left(\prod_{i=1}^Nw_i\right)(\det M_N)^2.
\end{equation}

The inverse is obtained by Lagrange interpolation.  Let
\[
 \Lambda_i(z)=\frac{Q_N(z)}{(z-x_i)Q_N'(x_i)}.
\]
The coefficient of $z^{j-1}$ in $\Lambda_i$ is
\[
 \frac{(-1)^{N-j}e_{N-j}(x_1,\ldots,\widehat{x_i},\ldots,x_N)}{Q_N'(x_i)}.
\]
Because $M_N$ carries one additional $x_i$ in its $i$th row,
\begin{equation}\label{eq:M-inverse}
 (M_N^{-1})_{ji}
 =\frac{(-1)^{N-j+i}w_i}{k_\rho}
 e_{N-j}(x_1,\ldots,\widehat{x_i},\ldots,x_N),
 \qquad 1\leq i,j\leq N.
\end{equation}
The block factorization~\eqref{eq:L-definition} then gives
\begin{equation}\label{eq:L-inverse}
 L_{N+1}^{-1}=
 \begin{pmatrix}
 1&0\\
 -M_N^{-1}\one&M_N^{-1}K_N^{-1}
 \end{pmatrix}.
\end{equation}

\begin{corollary}[Inertia]\label{cor:inertia}
For real nonzero $\rho$, the vector Gram form has the same inertia as
$k_\rho\operatorname{diag}_{0\leq i\leq N}((-1)^i)$.  The radical form has
the same inertia as
$\rho k_\rho\operatorname{diag}_{1\leq i\leq N}((-1)^i)$.
In particular, their positive and negative indices are obtained by counting
the two alternating sign strings; no zero eigenvalue occurs.
\end{corollary}

\begin{proof}
Congruence by an invertible real matrix preserves inertia.  The factors
$\lambda_i$ in the radical diagonal are all positive, so only the common
sign of $\rho k_\rho$ and the alternating factor remain.
\end{proof}

For comparison with Appendix~C of~\cite{Metsaev2024AdS}, the nodes at
$N=1,2,3$ are
\begin{equation}\label{eq:small-nodes}
 (2),\qquad (4,6),\qquad (6,10,12),
\end{equation}
and the weights are
\begin{equation}\label{eq:small-weights}
 (1),\qquad (3,2),\qquad (5,9,5).
\end{equation}
Together with the signs $(-1)^i$, these give exactly the displayed matrices
in dimensions six, eight, and ten.  Unlike an extrapolation from those three
cases, the proof of Sections~\ref{sec:spectral-polynomial}--\ref{sec:decoupling}
uses only the factorization of $Q_N$ and holds for every $N$.

\section{Functorial transport of the nonlinear structure}\label{sec:transport}

The preceding sections concern two bilinear forms.  The same matrices also
transport the nonlinear coefficient tensors, provided those tensors have
first been specified in the generic formulation.  We record the elementary
statement because it fixes the direction of every field and generator map.

Let $A$ be a finite-dimensional commutative algebra with basis $E_i$,
multiplication $E_iE_j=C_{ij}{}^kE_k$, and invariant form $H$.  For an
invertible matrix $L$, set $E'_a=L_{ai}E_i$.  Then
\begin{equation}\label{eq:transport-tensors}
 C'_{ab}{}^c=L_{ai}L_{bj}C_{ij}{}^k(L^{-1})_{kc},
 \qquad H'=LHL^{\mathsf T}.
\end{equation}

\begin{proposition}\label{prop:tensor-transport}
The transformation~\eqref{eq:transport-tensors} preserves commutativity,
associativity, and invariance of the bilinear form.  If $I\subset A$ is an
ideal with its own basis change $R$, it also preserves all structure maps
$A\otimes I\to I$ and $I\otimes I\to I$ and hence the corresponding extended
Yang--Mills gauge algebra.
\end{proposition}

\begin{proof}
For example, substitution in associativity gives
\[
 C'_{ab}{}^sC'_{sc}{}^d=C'_{bc}{}^sC'_{as}{}^d,
\]
because the adjacent $L^{-1}L$ factors contract to the identity.  The same
cancellation proves commutativity and
$H'(E'_aE'_b,E'_c)=H'(E'_a,E'_bE'_c)$.  For an ideal, insert $L$ on every
$A$ input, $R$ on every $I$ input, and the appropriate inverse on the output.
Every ideal, Jacobi, and invariance identity is again a contraction in which
the adjacent matrices cancel.
\end{proof}

In the present application $L=L_{N+1}$ and $R=K_NM_N$.  If the generic and
decoupled generators are collected in columns $\mathbf T,\widetilde{\mathbf
T}$ and $\mathbf S,\widetilde{\mathbf S}$, respectively, the Appendix~C
relations are
\begin{equation}\label{eq:generator-maps}
 \widetilde{\mathbf T}=L_{N+1}\mathbf T,
 \qquad
 \widetilde{\mathbf S}=K_NM_N\mathbf S.
\end{equation}
This is exactly the convention $E'=LE$ used in
Proposition~\ref{prop:tensor-transport}.

The connection, gauge parameter, and curvature are basis-independent elements
of the tensor-product gauge algebra.  If their coordinate columns in the
generic and decoupled bases are denoted by $(\phi,\xi,F)$ and
$(\varphi,\eta,\mathcal F)$, then
\begin{equation}\label{eq:vector-field-maps}
 \phi=L_{N+1}^{\mathsf T}\varphi,
 \qquad \xi=L_{N+1}^{\mathsf T}\eta,
 \qquad F=L_{N+1}^{\mathsf T}\mathcal F.
\end{equation}
The radical fields satisfy the corresponding four maps
\begin{equation}\label{eq:radical-field-maps}
 \begin{aligned}
 \phi_{\mathrm{sub}}&=M_N^{\mathsf T}K_N\varphi_{\mathrm{sub}},&
 F_{\mathrm{sub}}&=M_N^{\mathsf T}K_N\mathcal F_{\mathrm{sub}},\\
 \xi_{\mathrm{sub}}&=M_N^{\mathsf T}K_N\eta_{\mathrm{sub}},&
 \phi_{\mathrm{sub}}^{A}&=M_N^{\mathsf T}K_N\varphi_{\mathrm{sub}}^{A}.
 \end{aligned}
\end{equation}
Applying~\eqref{eq:transport-tensors} to
$F=d\phi+\frac12[\phi,\phi]$ and
$\delta_\xi\phi=d\xi+[\phi,\xi]$ proves that these maps intertwine the gauge
transformations.

For $N=1,2,3$, direct transport of the curved products displayed in
\cite{Metsaev2024AdS} reproduces the published decoupled product tables.  At
$N=3$, in particular,
\[
 E_1^2=\frac43E_2+\frac{28}{3}\rho E_3,
 \qquad E_1E_2=E_3.
\]
The curvature term is essential.  Proposition~\ref{prop:tensor-transport}
does not replace it by a flat monomial product; it says that whichever source
tensor is present is carried through the already proved invertible map.

\section{What the decoupling determines}
\label{sec:quadratic-determination}

The congruences of Theorem~\ref{thm:main} also have a converse at the
level of quadratic forms.  This observation is elementary, but it separates
two questions that are otherwise easy to conflate: the quadratic action is
fixed by its diagonalization, whereas the nonlinear coefficient algebra is
not.

Write
\begin{align*}
 \Delta_N&=k_\rho\operatorname{diag}(1,-1,\ldots,(-1)^N),\\
 \Delta_N^{\rm sub}
 &=k_\rho\operatorname{diag}_{1\leq i\leq N}((-1)^i x_i),
 \qquad B_N=K_NM_N.
\end{align*}

\begin{proposition}[Uniqueness of the quadratic pullback]
\label{prop:quadratic-pullback}
For every $N\geq1$ and $\rho\neq0$,
\begin{equation}\label{eq:unique-pullbacks}
 H_N=L_{N+1}^{-1}\Delta_NL_{N+1}^{-\mathsf T},\qquad
 H_N^{\rm sub}=B_N^{-1}\Delta_N^{\rm sub}B_N^{-\mathsf T}.
\end{equation}
Consequently, once the changes of basis and the diagonal normalizations are
fixed, the two generic quadratic forms are uniquely determined.
\end{proposition}

\begin{proof}
The matrices $L_{N+1}$ and $B_N$ are invertible by
Theorem~\ref{thm:main}.  Multiplying
\eqref{eq:intro-vector-congruence} and
\eqref{eq:intro-radical-congruence} by their inverses on the left and on the
right gives~\eqref{eq:unique-pullbacks}.
\end{proof}

Proposition~\ref{prop:quadratic-pullback} is an unconditional statement of
linear algebra.  Its interpretation as an action in the generic curvature
variables additionally requires an independent derivation of the field map
to the diagonal free multiplets.

Once the proposed action and pairing are assumed, the passage to the density
formula of~\cite{Metsaev2024AdS} is a finite reindexing.

\begin{proposition}[Conditional finite regrouping]
\label{prop:conditional-regrouping}
Assume that the action and pairing Ansatz (2.41)--(2.42) of
\cite{Metsaev2024AdS} holds in dimension $D=2N+4$.  Then the corresponding
generic curvature action has the finite regrouping stated in equation
(2.43) of that paper.
\end{proposition}

\begin{proof}
After substituting the pairing coefficients, the vector contribution is
\begin{equation}\label{eq:vector-generic-sum}
 -\frac14
 \sum_{\substack{0\leq m,n\leq N\\m+n\leq N}}
 G_{m+n}(-\rho)^{N-m-n}F_{2m+2}F_{2n+2},
\end{equation}
and the radical contribution is
\begin{equation}\label{eq:radical-generic-sum}
 -\frac12
 \sum_{\substack{1\leq m,n\leq N\\m+n\leq N+1}}
 G_{m+n-1}(-\rho)^{N+1-m-n}F_{2m+1}F_{2n+1}.
\end{equation}
In~\eqref{eq:vector-generic-sum} set $k=m+n$; in
\eqref{eq:radical-generic-sum} set $k=m+n-1$.  These substitutions are
bijections from the two displayed index sets to, respectively,
\[
 0\leq k\leq N,\quad m+n=k,
 \qquad
 1\leq k\leq N,\quad m+n=k+1.
\]
Both terms therefore carry the same factor $-G_k(-\rho)^{N-k}$, and their
sum is
\begin{equation}\label{eq:conditional-density-regrouping}
 -\sum_{k=0}^{N}G_k(-\rho)^{N-k}
 \left(
  \frac14\sum_{m+n=k}F_{2m+2}F_{2n+2}
  +\frac12
   \sum_{\substack{m,n\geq1\\m+n=k+1}}
   F_{2m+1}F_{2n+1}
 \right).
\end{equation}
The second inner sum is empty when $k=0$.  Thus
\eqref{eq:conditional-density-regrouping} is exactly the density definition
in (2.43).  This proves the finite regrouping (2.43); it does not prove the
Ansatz (2.41)--(2.42).
\end{proof}

The remaining issue is genuinely nonlinear.  A simple calculation at the
first dimension not worked out in~\cite{Metsaev2024AdS} shows that the
quadratic form and the most immediate algebraic constraints do not select a
unique product.

\begin{proposition}[A one-parameter family in the fixed generic basis]
\label{prop:n4-nonuniqueness}
For $\rho\neq0$ and $N=4$, the algebraic constraints consisting of
associativity, Frobenius invariance, the Gram form $H_4$, and a fixed
$\rho=0$ specialization admit a one-parameter family of pairwise distinct
products in the fixed generic basis.
\end{proposition}

\begin{proof}
Let multiplication by the first radical generator be represented in the
fixed generic basis by
\begin{equation}\label{eq:n4-family}
 A_\tau=
 \begin{pmatrix}
 0&0&0&0&0\\
 1&0&0&0&0\\
 0&\frac32&0&0&0\\
 0&\rho(\frac{109}{10}+\frac{\tau}{60})&\frac32&0&0\\
 0&\tau\rho^2&30\rho&1&0
 \end{pmatrix},
 \qquad \tau\in\mathbb R.
\end{equation}
Direct substitution gives
\begin{equation}\label{eq:n4-self-adjoint-family}
 A_\tau^{\mathsf T}H_4=H_4A_\tau.
\end{equation}
The four entries immediately below the diagonal are
$1,\frac32,\frac32,1$.  Hence
\[
 1,A_\tau1,\ldots,A_\tau^4 1
\]
is a basis for every $\tau$.  Identifying this cyclic basis with
$1,t,\ldots,t^4$ transports the product of
$\mathbb R[t]/(t^5)$ to the displayed generic basis.  The resulting product
is commutative and associative.  Equation
\eqref{eq:n4-self-adjoint-family} implies invariance of $H_4$ for every
polynomial in $A_\tau$, and therefore Frobenius invariance for the transported
product.

The curvature-dependent entries in~\eqref{eq:n4-family} vanish at $\rho=0$,
so the flat specialization is independent of $\tau$.  The flat-space
coefficients in equation (B.2) of~\cite{Metsaev2024AdS} are
\[
 e_{m,n}=\frac{(m+n)!}{m!n!}
 \frac{(N-m)!(N-n)!}{N!(N-m-n)!}.
\]
For $N=4$ they give
\[
 (e_{1,0},e_{1,1},e_{1,2},e_{1,3})
 =(1,\frac32,\frac32,1),
\]
which are precisely the four subdiagonal entries above.  Thus the common
flat specialization is exactly the flat-space product of
equation (B.2).
Distinct values of $\tau$ give distinct multiplication operators for the
fixed first generator and hence distinct products in the fixed basis.
\end{proof}

No claim is made that different values of $\tau$ define inequivalent theories
under arbitrary curvature-dependent field redefinitions.  The proposition
asserts non-uniqueness only in the fixed generic basis used for the source
coefficient tables.

Proposition~\ref{prop:n4-nonuniqueness} does not rule out an
all-dimensional interacting theory.  It shows instead that the Gram form,
the flat specialization, associativity, and Frobenius invariance are not
enough to determine such a theory uniquely.

\section{Three spectral models}\label{sec:spectral-models}

The same quadratic lattice and factorial weights occur in several classical
settings.  These models clarify why the normalization has a closed form, but
none is needed as an assumption in Theorem~\ref{thm:main}.

\subsection{The Johnson graph}
Let $\mathcal J_N$ be the combinatorial Laplacian of the Johnson graph
$J(2N,N)$.  Its vertices are the $N$-subsets of a $2N$-set, with adjacency
defined by intersection of size $N-1$.  The adjacency eigenvalue in the
$i$th constituent is $(N-i)^2-i$, and its multiplicity is
\begin{equation}\label{eq:johnson-multiplicity}
 m_i=\binom{2N}{i}-\binom{2N}{i-1},\qquad 0\leq i\leq N,
\end{equation}
where the second binomial coefficient is zero at $i=0$
\cite{DalfoEtAl2021}.  Since the graph is $N^2$-regular, the Laplacian
eigenvalues are
\begin{equation}\label{eq:johnson-eigenvalue}
 N^2-((N-i)^2-i)=i(2N+1-i)=\lambda_i.
\end{equation}
A factorial simplification shows $m_0=1$ and $m_i=w_i$ for $i>0$.
Consequently $L_{N+1}=D_NV_N$ is the signed,
multiplicity-normalized spectral evaluation map of $\rho\mathcal J_N$.
The alternating signs arise from the kinetic form, not from the positive
trace form of the graph.

\subsection{\texorpdfstring{A finite $SU(2)$ model}{A finite SU(2) model}}
The tensor power of $2N$ spin-$\frac12$ representations decomposes as
\begin{equation}\label{eq:su2-decomposition}
 (\mathbb C^2)^{\otimes 2N}
 \cong\bigoplus_{J=0}^N V_J^{\oplus m_J},
 \qquad
 m_J=\binom{2N}{N-J}-\binom{2N}{N-J-1}.
\end{equation}
Each $V_J$ has a one-dimensional zero-weight space.  On the total-weight-zero
subspace, the operator $N(N+1)I-\mathbf J^2$ therefore has eigenvalue
$N(N+1)-J(J+1)$ with multiplicity $m_J$.  Taking $i=N-J$ recovers
\eqref{eq:johnson-eigenvalue} and~\eqref{eq:johnson-multiplicity}.  This is the
Boolean-lattice $\mathfrak{sl}_2$ model behind the Johnson scheme
\cite{Feinsilver2011}; it is not a claim that the Yang--Mills fields carry
this $SU(2)$ action.

\subsection{Legendre--Stirling coefficients}
It is convenient to remove the curvature scale from the spectral polynomial
by setting
\[
 \widehat Q_N(z):=\prod_{i=1}^N(z-\lambda_i),
 \qquad Q_N(z)=\rho^N\widehat Q_N(z/\rho).
\]
Equation~\eqref{eq:legendre-product} identifies its coefficients with the
finite Legendre spectrum, and the Legendre operator gives the
finite-dimensional identity
\[
 \widehat Q_N(z)=(-1)^N\det\bigl((N(N+1)-z)I-\mathcal L\bigr)
\]
on $\mathbb R[t]_{\leq N-1}$.  The Johnson model independently supplies the
same spectral nodes together with the multiplicities $w_i$.  The Legendre
description controls the coefficient triangle, while the Johnson description
explains the normalization weights.  Their meeting in the Bezoutian
identity~\eqref{eq:bezout-kernel} is what makes the decoupling map both
explicit and dimension-independent.

\section{Discussion}\label{sec:discussion}

Lemma~\ref{lem:unique-hankel-continuation} characterizes the unique monic
Hankel continuation compatible with the conjectured mass nodes and the
evaluation pattern visible in the three explicit cases of Appendix~C in
\cite{Metsaev2024AdS}.  Theorem~\ref{thm:main} proves its Gram congruences
for arbitrary $N$.  Its Vandermonde form comes from root evaluation for a
Bezoutian, and the square-root factors are Johnson multiplicities.  The same
argument gives the inverse map and the kinetic signatures in closed form.

Once a generic coefficient algebra and radical ideal satisfy the source
identities,
Proposition~\ref{prop:tensor-transport} carries them to the decoupled basis.
This conditional transport is enough to recover all published nonlinear
tables.  The diagonal free action also fixes the generic quadratic forms, and
the density formula (2.43) follows by a finite reindexing if the proposed
interacting action and pairing are assumed.  However,
Proposition~\ref{prop:n4-nonuniqueness} shows that the most immediate
algebraic continuation principles do not determine the missing nonlinear
product.  Extending the interacting generic formulation beyond the published
source products therefore requires additional field-theoretic input.

There is a related geometric direction.  Higher conformal Yang--Mills
energies and currents have been constructed from conformally compact
Yang--Mills theory~\cite{GoverLatiniWaldronZhang2024}, and a compact universal
formula for the corresponding conformal Yang--Mills condition is known on
even-dimensional conformal manifolds~\cite{BlitzGoverKopinskiWaldron2026}.
Relating these geometric functionals to the auxiliary-field elimination of
\cite{Metsaev2024AdS}, with the normalization and density basis of its
equation (7.24), is a separate functional-matching problem.  To our knowledge,
such a matching has not been established.  A matching theorem would provide a
direct route to deriving or testing that formula without relying on further
dimension-by-dimension data.

\enlargethispage{3\baselineskip}

\bibliographystyle{amsplain}
\bibliography{references}

\end{document}